\documentclass[oribibl]{llncs}  
\usepackage[dvips]{graphicx}                 
\usepackage{verbatim} 
\usepackage{times}
\usepackage{url}
\usepackage{algorithm}
\usepackage{algorithmic}
\usepackage{thmtools}
\usepackage{thm-restate} 

\newtheorem{observation}{Observation}

\makeatletter  
 \renewcommand\subsubsection{\@startsection{subsubsection}{3}{\z@}%
                        {-18\p@ \@plus -4\p@ \@minus -4\p@}%
                        {8\p@ \@plus 4\p@ \@minus 4\p@}
                        {\normalfont\normalsize\bfseries\boldmath
                         \rightskip=\z@ \@plus 8em
 \pretolerance=10000 }}
\makeatother   

\title{Dynamic 2D Dictionary Matching in Small Space} 
\author{Shoshana Marcus \inst{1} 
\and Dina Sokol \inst{2} 
}

\institute{Simons Center for Quantitative Biology, \\
Cold Spring Harbor Laboratory,
Cold Spring Harbor, NY, 11724
\\ \email{smarcus@cshl.edu}
\and
Department of Computer and Information Science, \\
Brooklyn College of the City University of New York, Brooklyn, NY, 11210
\\ \email { sokol@sci.brooklyn.cuny.edu}
}

\begin{document}
\date{}
\maketitle

\begin{abstract}
The dictionary matching problem preprocesses a set of patterns and finds all occurrences of each of the patterns in a text when it is provided.
We focus on the dynamic setting, in which patterns can be inserted to and removed from the dictionary, without reprocessing the entire dictionary.
This article presents the first algorithm that performs \emph{dynamic} dictionary matching on two-dimensional data within small space.  
The time complexity of our algorithm is almost linear.  The only slowdown is incurred by querying the compressed self-index that replaces the dictionary.  
The dictionary is updated in time proportional to the size of the pattern that is being inserted to or removed from the dictionary.  
Our algorithm is suitable for rectangular patterns that are of uniform size in one dimension. 

\end{abstract}                                                                                                                                                                                                              
                                             
\section{Introduction}


In the \emph{dictionary matching problem}, the task is to identify a \emph{set} of patterns, called a dictionary, within a given text.  Applications for this problem include searching the World-Wide Web for specific keywords, scanning a file for virus signatures, and network intrusion detection. 
The problem also has applications in the biological sciences, such as searching through a DNA sequence for a set of motifs.
Dictionary matching generalizes to the two-dimensional setting.  Image identification software, which identifies smaller images in a large image based on a set of known images, is a direct application of dictionary matching on two-dimensional data.

In recent years, there has been a massive proliferation of digital data. Some of the main contributors to this data explosion are the World-Wide Web, next generation sequencing, and increased use of satellite imaging.  
Concurrently, industry has been producing equipment with ever-decreasing hardware availability.  Thus, researchers are faced with scenarios in which this data growth must be accessible to applications running on devices that have reduced storage capacity, such as mobile and satellite devices.
Hardware resources are more limited, yet the data sets continue to escalate in size.
The added constraint of performing efficient dictionary matching using little or no extra space is a challenging and practical problem. 

It is often the case that the dictionary of patterns will change over time. Efficient dynamic dictionary
matching algorithms support insertion of a new pattern to the dictionary and removal of a pattern
from the dictionary, e.g., \cite{AmiFar91, 
AmiFGGP94,
IduSch94a, 
AmiFIPS95,
SahVis96,  
ChanHLS07,
HLSTVit09}. 
They thereby eliminate the need to reprocess the entire dictionary and can
adapt to changes as they occur.

Idury and Schaffer developed a dynamic dictionary matching algorithm for rectangular patterns of different sizes.  Their algorithm uses working space proportional to the input size and requires more than linear running time.
The existing dynamic 2D dictionary matching algorithms for square patterns use linear working space and  incorporate $O(\log \ell)$ or $O(\log^2 \ell)$ slowdown in processing the text and in updating the dictionary \cite{AmiFIPS95, ChoiLam96, Gia95}.

The objective of this paper is to develop the first dynamic dictionary matching algorithm for two-dimensional data in the space-constrained environment.  The existing static succinct 2D dictionary matching algorithm with no slowdown \cite{NeuSok11} is not suitable for the dynamic setting.  It relies on the succinct 1D dictionary matching algorithm of Hon et al.\ \cite{HonKSTV10}, which does not readily admit changes to the dictionary.
In this paper, we extend the succinct 2D dictionary matching algorithm of \cite{NeuSokAlgo}, along with the improvements of \cite{NeuSok2DLW}, to the dynamic setting.  We develop a dynamic algorithm that meets the time and space complexities that were achieved in the static version of the algorithm.  
The dictionary is initially processed in time proportional to the size of the dictionary. Subsequently, a pattern is inserted or removed in time proportional to the single pattern's size.  
We  modify the \emph{witness tree} \cite{NeuSokAlgo} to form a dynamic data structure that meets the space and time complexities achieved by the static version.  The dynamic witness tree accommodates insertion or removal of any string in time proportional to the string's length.

We define \emph{Two-Dimensional Dynamic Dictionary Matching} (2D-DDM) as follows.

\noindent {\bf Initial Input:} A dictionary of $d$ patterns, $D=\{P_1, \ldots, P_d\}$ and a text $T$ of size $n_1 \times n_2$. Each $P_i$ is of size $m_i \times \overline{m}$, with total size $|D|=\ell$. 

\noindent {\bf Update Dictionary:} Insert or remove a given pattern $P$, of size $p \times \overline{m}$.  

\noindent {\bf Process Text:} Find all occurrences of $P_i$, $1 \leq i \leq d$,  in $T$.

In this paper we present a time and space efficient algorithm to solve 2D-DDM.
During the preprocessing stage, our algorithm replaces the dictionary $D$ with a compressed self-index. We use $\tau$ to denote the time it takes to access a character or perform other queries in the compressed self-index of the dictionary;  using recent results $\tau$ is at most $\log^2 \ell$. The initial  preprocessing of the dictionary completes in $O(\ell\tau)$ time.  
A pattern $P$, of size $p \times \overline{m}$, is inserted to or removed from the dictionary in $O(p\overline{m}\tau)$ time.  
Our algorithm searches the text $T$ in $O(n_1n_2 \tau)$ time. The extra space used by our algorithm is $O(d\overline{m} \log d\overline{m}+dm' \log dm')$ bits, where $m'=max \{m_1,\ldots, m_d\}$.

The succinct 2D (static) dictionary matching algorithm of \cite{NeuSokAlgo} was presented in terms of a dictionary in which all patterns are the same size in both dimensions, resulting in a dictionary of size $dm^2$.  In this paper, we generalize this result to deal with a dictionary of patterns that are of uniform width, but of varying heights.  We perform a detailed analysis and distinguish between the sources of time complexities.  Specifically, we analyze which time complexities are proportional to the uniform width of the patterns, $\overline{m}$, which are proportional to the height\footnote{We chose this notation since it is visual.  The bar represents a uniform width, while the prime is vertical, representing a uniform height.} of the largest pattern, $m'$, and which are proportional to the actual dictionary size, $\ell$.  While doing this, we discovered the need for more efficient techniques in the verification process in order for the text scanning to remain linear in the size of the text, in the case of varying pattern heights. 
Herein lies one of the contributions of this paper.

We begin by presenting related work on 1D dynamic dictionary matching in Section \ref{sec:1D-DDM}. Then, in the following section, we present a linear-time dynamic 2D dictionary matching algorithm that uses extra space proportional to the size of the input.  
In Section \ref{sec:dgeqm}, we describe a succinct variation of this linear space algorithm for a dictionary with a large number of patterns.  
For dictionaries in which the number of patterns is small relative to their size, we describe our approach in Section \ref{sec:dleqm}. 
We conclude with open problems in Section \ref{sec:conclusion}.

\section{1D Dynamic Dictionary Matching}\label{sec:1D-DDM}
In this section we summarize the existing algorithms for 1D dynamic dictionary matching since we build our two-dimensional algorithm on one-dimensional algorithms.  
The 1D dictionary consists of $d$ one-dimensional patterns of total size $\ell'$, drawn from an alphabet of size $\sigma$.

The first dynamic dictionary matching algorithms use suffix trees and incur an $O(\log \ell')$ slowdown in runtime \cite{AmiFar91, AmiFGGP94}.
Idury and Schaffer \cite{IduSch94a} developed a dynamic version of the classic Aho-Corasick automaton \cite{AC} in which the dictionary is preprocessed in linear time.  However, the tasks of updating the dictionary and scanning text require extra time. 
The culmination of work by separate groups of researchers on dynamic dictionary matching \cite{AmiFar91,AmiFGGP94,IduSch94a}
is an algorithm that mimics the Aho-Corasick automaton but stores the \textit{goto} and \textit{report} transitions separately \cite{AmiFIPS95}.  The time complexity of this algorithm is close to linear, albeit with an $O(\frac{\log \ell'}{\log \log \ell'})$ slowdown to update the dictionary or to scan text. 

Sahinalp and Vishkin achieved dynamic dictionary matching with no slowdown \cite{SahVis96}.  
The preprocessing time of their algorithm is linear in the size of the dictionary, text scanning is linear in the size of the text, and the dictionary is updated in time proportional to the size of the pattern being added or removed.  The time complexity of this algorithm meets the standard set by Aho and Corasick for static dictionary matching. 

Sahinalp and Vishkin's algorithm relies on compact tries and an original data structure called a fat tree.  
Their algorithm employs a naming technique and identifies cores of each pattern using a compact representation of the fat tree.  
If a pattern matches a substring of the text, then the main core of the pattern and the text substring will necessarily be aligned.  Conversely, if the main cores do not match, the text is easily filtered to a limited number of positions at which a pattern can occur. Dictionary patterns are classified into groups according to the level of their main core.  Then, an independent data structure is built for each group, which consists of two compact tries.  In total, this algorithm uses working space proportional to the size of the input.

For dynamic dictionary matching in the space-constrained application, Chan et al.\  use the compressed suffix tree for succinct dictionary matching \cite{ChanHLS07}.  They build on the ideas of Amir and Farach \cite{AmiFar91} to use the suffix tree for dictionary matching.  They replace the suffix tree with a compressed suffix tree developed by Sadakane \cite{Sad07}, which is stored in $O(\ell')$ bits, and show how to make the data structure dynamic.  They describe how to answer lowest marked ancestor queries by a balanced parenthesis representation of the  nodes.  The time complexity of inserting and removing a pattern and of scanning text has a slowdown of $O(\log^2 \ell')$. 

An improved succinct dynamic dictionary matching algorithm was developed by  Hon et al. \cite{HLSTVit09}.  It uses space that meets $k$th order empirical entropy bounds of the dictionary, $\ell' H_k(D)+o(\ell' \log \sigma)+O(d \log \ell')$ bits of space.  The suffix tree is sampled to save space and an innovative method is proposed for a lowest marked ancestor data structure.  
They introduce the combination of a dynamic interval tree with a Dietz and Sleator order-maintenance data structure as a framework for answering lowest marked ancestor queries efficiently. 
Inserting or removing a dictionary pattern $P$, of length $p$, requires $O(p \log \sigma+\log \ell')$ time and searching a text of length $n$ requires $O(n \log \ell' + occ)$  time.

$H_k(S)$, i.e., the $k$th order empirical entropy of a string $S$, describes the minimum number of
bits that are needed to encode each symbol of the string within context, and it is often used to demonstrate that storage space meets the information-theoretic lower bounds of data.


\section{2D-DDM in Linear Space}\label{sec:dynamicBB} 

In this section we present a linear-time dynamic 2D dictionary matching algorithm that uses extra space proportional to the size of the input.  The first linear-time 2D \emph{single} pattern matching algorithm was developed independently by Bird \cite{Bird} and by Baker \cite{Baker}.  They translate the 2D pattern matching problem into a 1D pattern matching problem.  Rows of the pattern are perceived as metacharacters and named so that distinct rows receive different names.  The text is named in a similar fashion and 1D pattern matching is performed over the text columns and the pattern of names.

The Bird / Baker algorithm readily extends to dictionary matching by replacing the 1D single pattern  matching mechanism, a Knuth-Morris-Pratt automaton, with 1D dictionary matching, an Aho-Corasick automaton.  In the multiple pattern matching version of the Bird / Baker algorithm, 1D dictionary matching is used in two different ways.  First, the pattern rows are seen as a 1D dictionary and this set of ``patterns'' is used to linearize the dictionary and then to label text positions.  A separate 1D dictionary is formed of the linearized 2D patterns.  The Bird / Baker algorithm is suitable for 2D patterns that are of uniform size in at least one dimension, so that the text can be marked with at most one name at each text location.  The Bird / Baker  method uses linear time and space in both the pattern preprocessing and the text scanning stages.  

Sahinalp and Vishkin's \cite{SahVis96} dynamic 1D dictionary matching algorithm (SV) uses a naming technique rather than a dictionary-matching automaton.  Yet, it is a suitable replacement for the Aho-Corasick automata in the Bird / Baker algorithm.  Thus, the combination of these techniques, one for dynamic dictionary matching in 1D and the other for static 2D dictionary matching, yields a dynamic 2D dictionary matching algorithm that runs in linear time.  This modification extends the Bird / Baker algorithm to accommodate a changing dictionary, yet it does not introduce any slowdown.  We outline this process in Algorithm \ref{alg:dynamicBB}.

	\begin{algorithm}
		\caption{Dynamic Version of Bird / Baker Algorithm}
		\label{alg:dynamicBB}
	\begin{algorithmic}
	\STATE \COMMENT 1  Preprocess Pattern: \\
		 	  \hspace{\algorithmicindent}  a) Name pattern rows using SV \cite{SahVis96}.   \\ 
		 	  \hspace{\algorithmicindent}  b) Store 1D pattern of names for each pattern in $D$, called $D'$. \\ 
		 	  \hspace{\algorithmicindent}   c) Preprocess $D'$ using SV to later perform 1D dynamic dictionary matching. 
		\STATE \COMMENT 2  Row Matching: \\
			\hspace{\algorithmicindent}  Use SV on each row of text to find occurrences of $D$'s pattern rows.  \\ 
			\hspace{\algorithmicindent} This labels positions at which a pattern row ends.
		 \STATE \COMMENT 3 Column Matching: \\
			\hspace{\algorithmicindent} 	Run SV on named columns of text to find occurrences of patterns from $D'$ in the text. \\
			\hspace{\algorithmicindent} 	Output pattern occurrences.

	\end{algorithmic}
	\end{algorithm}

Initially, the dictionary of pattern rows is empty.  One 2D pattern is linearized at a time, row by row.  As a pattern row is examined, it can be viewed as a text on which to perform dictionary matching.  If a pattern row is identified in the new pattern row, then it is given the same name as the matching row.  Otherwise, this new row is seen as a new 1D pattern and added to the dictionary of pattern rows.
Once the pattern rows have been given names, the 1D patterns of names, $D'$, are preprocessed separately. 
		
		Whenever a pattern is added to or removed from the 2D dictionary, the precomputed information about the patterns can be adjusted in time proportional to the size of the 2D pattern that is entering or leaving the dictionary.  That is, Sahinalp and Vishkin's framework for dictionary matching allows both 1D dictionaries to efficiently react to a change in the 2D linearized dictionary that they represent.

		\textbf{Space complexity of Algorithm \ref{alg:dynamicBB}:}
		The dynamic version we present of the Bird / Baker algorithm uses extra space proportional to the size of the input.  It uses $O(\ell \log \ell)$ bits of extra space to name the pattern rows using SV \cite{SahVis96} and $O(dm' \log dm')$ bits of extra space to store and index the 1D representation of the patterns.  
During text scanning, $O(n_2 \log n_2)$  bits of space are used to run SV on each row of text and  $O(n_1 \log n_1)$ bits of space are used to run SV on the named columns of text, one at a time. $O(n_1n_2 \log dm')$ bits of extra space are used to store the names given to text positions. 
		
\section{2D-DDM in Small-Space For Large Number of Patterns}\label{sec:dgeqm}
The dynamic version of the Bird / Baker algorithm presented in Section \ref{sec:dynamicBB} uses linear working space.  
In this section we present a variation of Algorithm \ref{alg:dynamicBB} that runs in small space for a dictionary in which $d \geq \overline{m}$.  That is, when the number of patterns is larger than the width of a pattern.

We begin by modifying Algorithm \ref{alg:dynamicBB} to work with small blocks of text and thereby relate the extra space to the size of the dictionary, not the size of the text.
We use a known technique for minimizing space and process the text in small overlapping blocks of size $3m'/2 \times 3\overline{m}/2$.  
Since each text block is processed in time proportional to the size of the text block, the overall text scanning time remains linear.

By processing one text block at a time, we reduce the working space to $O(\ell \log \ell+dm' \log dm')$ bits of extra space to preprocess the patterns and $O(\overline{m} \log \overline{m}+ \overline{m}m' \log dm')$ 
bits of extra space to search the text. 
This change does not affect the time complexity.  We seek to further reduce the working space by employing a smaller space mechanism to name the pattern rows and subsequently name the text positions.


Recent innovations in succinct full-text indexing provide us with the ability to compress a suffix tree,
using space that is proportional to the entropy of the original data it is built upon.
These self-indexes can replace the original text, as they support retrieval of the original text,
in addition to answering queries about the data, very quickly.

Several dynamic compressed suffix tree representations have been developed, each offering a different time/space trade-off.
Chan et al.\  presented a dynamic suffix tree that occupies $O(\ell)$ bits of space \cite{ChanHLS07}.  Queries, such as edge label retrieval and insertion or removal of a substring, have an $O(\log ^2 \ell)$ slowdown.
Russo et al.\  developed a dynamic fully-compressed suffix tree requiring  $\ell H_k(\ell)+o(\ell \log \sigma)$ bits of space, which is 
asymptotically optimal under $k$th order empirical entropy \cite{RusNavOli11}.  This compressed suffix tree representation uses a dynamic compressed suffix array and stores a sample of the suffix tree nodes.  Although some operations can be executed more quickly, all operations have $O(\log^2 \ell)$ time complexity. 
This dynamic compressed suffix tree supports a larger set of suffix tree navigation operations than the compressed suffix tree proposed by Chan et al. \cite{ChanHLS07}.  It also reaches a better space complexity and can perform basic operations more quickly.  We hereafter suppose that a dynamic compressed suffix tree is used to replace the dictionary of patterns and we refer to the slowdown of operations in the entropy-compressed self-index as $\tau$. 

We now describe a succinct version of Algorithm \ref{alg:dynamicBB} that uses a dynamic compressed suffix tree to represent and index the pattern rows in entropy-compressed space.  Its  modifications are limited to steps 1a and 2 in Algorithm \ref{alg:dynamicBB}.  
Traversing the dynamic compressed suffix tree introduces $\tau$ slowdown in running time.
During pattern preprocessing, the dynamic compressed suffix tree is built incrementally, as each pattern row is named.
First, traversal of the suffix tree is attempted by traversing a path from the root labeled by the characters in the pattern row.  If a matching row is found, the new row is given the same name as the row that it matches.  Otherwise, the new pattern row is inserted into the compressed suffix tree and given a new name.

The positions of a text block row are also named by traversing the suffix tree.  Here the suffix tree is not modified by the text.  
We use a technique similar to the one described by Gusfield in the computation of \emph{matching statistics}, \cite{Gusfield97} Section 7.8.  Positions in a text block are named, row by row, according to the names of pattern rows. To name a new text block row, traversal begins at the root of the tree, with the edge whose label matches the first position of the text block row. When $\overline{m}$ consecutive characters trace a path from the root, traversal reaches a leaf, and the position is named with the matching pattern row. At a mismatch, suffix links quickly find the longest suffix of the already matched string that matches a prefix of some pattern row and the next text character is compared to that labeled edge of the tree. 

All pattern rows have width $\overline{m}$.  This ensures that each text position can be uniquely labeled.
One pattern row cannot be a substring of another.  Thus, we do not share the concern of Amir and Farach's suffix tree based approach to dictionary matching \cite{AmiFar91}.  They use lowest marked ancestor queries to address the issue of possibly skipping over pattern occurrences in the case that one pattern is a substring of another and a suffix link is traversed.

\begin{restatable}{theorem}{thmComplexityDgeqm}
\label{thm:complexity_dgeqm}
If $d \geq \overline{m}$, we can solve the dynamic 2D dictionary matching problem in almost linear $O((\ell+n_1n_2) \tau)$ time and 
$O(\overline{m} \log \overline{m}+dm' \log dm')$ bits of extra space, aside from the space used to represent the dictionary in a compressed self-index.  Pattern $P$ of size $p \times \overline{m}$ can be inserted to or removed from the dictionary in $O(p\overline{m}\tau)$ time and the updated index will occupy an additional $O(p \log dm')$ bits of space, where $m'$ is updated to reflect the new maximum pattern height.
\end{restatable}

\vspace{-10pt}
\section{2D-DDM in Small-Space for Small Number of Patterns}\label{sec:dleqm}
This section deals with the case in which the number of patterns is smaller than the common dimension among all dictionary patterns, i.e., $d=o(\overline{m})$. For this case, we do not allow the space to label each text block location and therefore the dynamic version of the Bird and Baker algorithm cannot be applied trivially. We present several combinatorial tricks to preserve the spirit of Bird and Baker's algorithm without incurring the necessary storage overhead.  
The dictionary is indexed by a dynamic compressed suffix tree, after which the patterns can be discarded.  This can be done in space that meets $k$th order empirical entropy bounds of the input, as described in Section \ref{sec:dgeqm}. Thus, the compressed self-index does not occupy \textit{extra space}.
Throughout this section, the extra space used by our algorithm is limited to $O(\overline{m} \log \overline{m}+dm' \log dm')$ bits of space.  The running time of our algorithm is almost linear, with a slowdown to accommodate queries to the compressed suffix tree, referred to as $\tau$.

We divide the dictionary patterns into two groups and search the text for patterns in each group separately.  
In the following sections, we describe first an algorithm for patterns in which the rows are highly periodic and then an algorithm for all other patterns.  We begin by describing a dynamic data structure that is used by both parts of the algorithm.



\subsection{Dynamic Witness Tree}\label{sec:dynamicDS}
In this section we show how to form a dynamic variant of the witness tree, a data structure that was introduced in \cite{NeuSokAlgo}.  
Given a set $S$ of $j$ strings, each of length $m$, a witness tree can be constructed to name these strings in linear $O(jm)$ time and in $O(j \log j)$ bits of space so that identical strings receive the same name \cite{NeuSokAlgo}.
An internal node in the witness tree denotes a position of mismatch, which is an integer $\in$ [1, $m$].	
Each edge of the tree is labeled with a single character.  Sibling edges must have different labels.
A leaf represents a name given to string(s) in $S$. 

\noindent{\bf Query:}  For any two strings $s, s' \in S$, return a position of mismatch between $s$ and $s'$ if $s \neq s'$, otherwise return $m+1$.

Preprocessing the witness tree for Lowest Common Ancestor (LCA) queries on its leaves allows us to answer the above witness query between any two named strings in $S$ in constant time.  This preprocessing can be performed in linear time and space, with respect to the size of the tree, even for a  dynamically changing tree \cite{ColHar05}.

Construction of the witness tree begins by choosing any two strings in $S$ and comparing them sequentially. When a mismatch is found, comparison halts and an internal node is created to represent this witness of mismatch, with two children to represent the names of the two strings.  If no mismatch is found, the two strings are given the same name.  Each successive string is compared to the witnesses stored in the tree by traversing a path from the root to identify to which name, if any, the string belongs.
Characters of a new string are examined in the order dictated by traversal of the witness tree, possibly out of sequence.  If traversal halts at an internal node, the string receives a new name, and a new leaf is added as a child to the internal node.  Otherwise, traversal halts at a leaf, and the new string is compared sequentially to the string represented by the leaf, as done with the first two strings.  

Now we consider the scenario in which $S$ is a dynamically changing set of strings.

\begin{lemma}\label{lemma:witTreeAdd}
A new string is added to the witness tree in $O(m)$ time.  
\end{lemma}

\begin{proof}
Including a new string in $S$ and naming it with the witness tree follows the same procedure that the static witness tree uses to build the witness tree as each pattern is considered individually.  This is done in $O(m)$ time and adds one or zero nodes to the witness tree \cite{NeuSokAlgo}.
\qed
\end{proof}

\begin{lemma}\label{lemma:witTreeRemove}
A string is removed from the witness tree in $O(1)$ time.  
\end{lemma}

\begin{proof}
In removing a string $s$ from $S$, there are two possibilities to consider.  If $s$ is the only string with its name, remove its leaf.  In the event that the parent is an internal node with only one other child, remove the hanging internal node as well.  Then, the sibling of the deleted leaf becomes a child of its grandparent.  The other possibility is that some other string(s) in $S$ bear the same name as $s$.  We do not want to remove a leaf while there is still a string in $S$ that has its name.  Thus, we augment each leaf with an integer field to store the number of strings in $S$ that have its name.
This counter is increased when a new string is named with an existing name. This counter is decreased when a row is deleted.  When the counter is down to 0, the leaf is discarded, possibly along with its parent node, as described earlier.
\qed
\end{proof}

\begin{observation}\label{obs:witTreeSize}
The dynamic witness tree of $j$ strings, each of length $m$, occupies $O(j \log j)$ bits of space. 
\end{observation}

\subsection{Group I Patterns}\label{sec:group1alg}
A string $S$ is \emph{primitive} if it cannot be expressed in the form $S=u^j$, for  $j > 1$ and any prefix $u$ of $S$.
String $S$ is \emph{periodic} in $u$ if $S=u^ju'$ where $u'$ is a prefix of $u$, $u$ is primitive, and $j \geq 2$.  
A periodic string can be expressed as $u^ju'$ for one unique primitive $u$.
We refer to $u$ as ``the period'' of $p$. Depending on the context, we use the term \emph{period} to refer to either the string $u$ or the period size $|u|$.

There are two types of patterns, and each one presents its own difficulty.  In the initial preprocessing step, we divide the patterns into two groups based on the 1D periodicity of their rows.  In Group I, all pattern rows are periodic, with periods $\leq \overline{m}/4$. The difficulty in this case is that many overlapping occurrences can appear in the text in close proximity to each other,
and we can easily have more candidates than the working space we allow.  Patterns in Group II have at least one aperiodic row or one row whose period is larger than $\overline{m}/4$.  Here, each pattern can occur only $O(1)$ times in a text block. Since several patterns can overlap each other in both directions, a difficulty arises in the text scanning stage.  We do not allow the time to verify different candidates separately, nor do we allow space to keep track of the possible overlaps between different patterns. 


\subsubsection{Preprocessing Dictionary} \label{sec:group1pre}
For patterns in Group I, we linearize the patterns with \emph{Lyndon word naming} \cite{NeuSokAlgo} on the  rows.
Two strings are conjugate if they differ only by a cyclic permutation of their characters.
A Lyndon word is a primitive string which is the smallest of its conjugates for the alphabetic ordering.
Lyndon word naming classifies strings by the conjugacy of their periods and uses the Lyndon word as the class representative.  
Once Lyndon word naming has been performed, each pattern row is represented by the name of its period's class and its
\emph{LWpos}, the first position at which the Lyndon word begins in the row.  We use the dynamic witness tree to perform Lyndon word naming in linear time.  





A pattern occurs in a text block if the 1D representations are the same and the periods align within each row.
The \emph{2D Lyndon word} is a succinct representation of the Lyndon word that is conjugate to each row's period combined with the relative alignments of the Lyndon words among the matrix rows.
2D Lyndon word naming forms equivalence classes of patterns with the same 1D name and uses the 2D Lyndon word in each class as the class representative.  
The 2D Lyndon word that represents an $m_i \times \overline{m}$ matrix is computed in sublinear time and $O(m_i \log m_i)$ bits of working space \cite{NeuSok2DLW}.  We classify the patterns with 2D Lyndon word naming so that the text scanning stage can efficiently verify patterns occurrences.

The distance between any two overlapping occurrences of $P_i$ in the same row is the Least Common Multiple (LCM) of the periods of all rows of $P_i$.  We precompute the LCM of each pattern so that $O(1)$ space suffices to store all occurrences of a  pattern in a row, and $O(dm' \log dm')$ bits of space suffice  to store all patterns occurrences.  The LCM is computed incrementally, row by row.  
The LCM table stores the LCM of the periods of the first $i$ rows of the pattern as $LCM[i]$, for $1 \leq i \leq m_i$, and is available during text scanning.  
Although the LCM can be exponential in $m'$, we only need the elements of the LCM table that are polynomial in $m'$, as discussed in \cite{NeuSok2DLW}.
 
The following preprocessing steps are initially performed for each dictionary pattern in Group I and are later used upon arrival of a new pattern.

\begin{enumerate}	
	\item \label{pre1perrow} For each pattern row,
		\begin{enumerate}	
	\item Compute period and canonize. 
	\item Lyndon word naming with dynamic witness tree, resulting in 1D dictionary $D'$.
	\item Insert to dynamic compressed suffix tree. 
		\end{enumerate}	
	\item \label{pre1dpats} Preprocess 1D dictionary:
		\begin{enumerate}
	\item Preprocess $D'$ for dynamic dictionary matching.
	\item Build LCM table for each 1D pattern.
	\item For each linearized pattern whose 1D form is not periodic or if $m' =O (\overline{m})$: \\
	 Compute 2D Lyndon word and column $z$ it occurs in.   
	 \item For each linearized pattern whose 1D form is periodic when $\overline{m} =o(m')$: 
	 		\begin{enumerate}
	 			\item Compute 2D Lyndon word and the column $z$ it occurs in for each $p\_block$, a period in the 1D pattern.  
	 			\item Classify $p\_block$s by  2D Lyndon word naming.  
	 			\item Compute the difference between $z$ in adjacent $p\_blocks$.
	 			\item Build KMP automaton for named $p\_block$s and the differences between their $z$ values.
	 		\end{enumerate}

		\end{enumerate}
\end{enumerate}

\begin{restatable}{lemma}{lemmaGroupIpre}
\label{lemma:group1pre}
Patterns in Group I are preprocessed in $O(\ell \tau)$ time and $O(\overline{m} \log \overline{m}+dm' \log dm')$ bits of extra space.
\end{restatable}

\begin{proof}
Step \ref{pre1perrow} processes a single pattern row in $O(\overline{m} \tau)$ time and $O(\overline{m} \log \overline{m})$ bits of extra space \cite{NeuSokAlgo}.  Thus, the entire set of pattern rows are processed in $O(\ell)$ time to gather information and $O(\ell \tau)$ time to index the pattern rows in a dynamic compressed suffix tree. 
Since $O(1)$ information is stored per row, $O(dm' \log dm')$ bits of extra space are used to store information gathered about the pattern rows in the dictionary.  \\
Step \ref{pre1dpats} preprocesses the 1D patterns in the dictionary of names.  Using Sahinalp and Vishkin's algorithm, $O(dm')$ time and $O(dm' \log dm')$ bits of extra space are used to facilitate linear time dynamic dictionary matching in a 1D dictionary of size $O(dm')$ \cite{SahVis96}.  
The LCM tables of the 1D patterns are computed in linear time and occupy $O(dm' \log m')$ bits of extra space.
The 2D Lyndon word of each pattern is computed in sublinear time with respect to its size and the set of 2D Lyndon words occupy $O(dm' \log dm')$ bits of extra space \cite{NeuSok2DLW}. 
Similarly, $p\_block$s are classified by their representative 2D Lyndon words in sublinear time and a KMP automaton of the $p\_block$s is constructed in $O(m')$ time \cite{KMP}.
Overall, Step \ref{pre1dpats} runs in $O(dm')$ time and $O(dm' \log dm')$ bits of extra space.
\qed
\end{proof}

\begin{corollary}\label{cor:group1add}
A new pattern of size $p \times \overline{m}$ is added to Group I in
$O(p\overline{m} \tau)$ time and $O(\overline{m} \log \overline{m}+p \log dm')$ bits of extra space. 
\end{corollary}



\begin{restatable}{lemma}{lemmaGroupIremove}
\label{lemma:group1remove}
A pattern in Group I of size $p \times \overline{m}$ is removed from the dictionary in $O(p\overline{m} \tau)$ time and eliminates $O(\overline{m} \log \overline{m}+p \log dm')$ bits of extra space the algorithm allocated for it. 
\end{restatable}

\begin{proof}
The following steps meet the indicated time and space bounds and remove a pattern from Group I.  
Each pattern row is removed from the dynamic witness tree, in $O(1)$ time  (by Section \ref{sec:dynamicDS})
,  and from the dynamic compressed suffix tree, in $O(\overline{m} \tau)$ time.  This takes $O(p \overline{m} \tau)$ time in total. 
If this is the only pattern with its 1D representation, its LCM table is deleted and the 1D pattern is removed from the dictionary of names that has been preprocessed for dynamic dictionary matching.  
If this is one of several patterns with the same 1D representation, and the sole member of its consistency class, its representative 2D Lyndon word is removed from the compressed trie.
\qed
\end{proof}         

\subsubsection{Text Scanning}

The text is searched for occurrences of patterns in Group I in a three step process.  First, Lyndon word naming is performed on the rows of the text block  using the dynamic witness tree of the dictionary (Section~\ref{sec:dynamicDS}).  We store the name of its period's class, period size, \emph{LWpos}, \emph{right}, and \emph{left} of each pattern row.   
Then, the linearized text, $T^{\prime}$, is searched for candidate positions that match a pattern in the 1D dictionary using 1D dynamic dictionary matching, since the patterns can be of varying heights.  
Finally, the verification step finds the actual pattern occurrences among the candidates.   Since the first two steps have been described, the remainder of this section discusses the verification stage.

To verify candidates, we consider the alignment of periods among rows and the overall width of the 1D names in the text block. 
If $m' =O (\overline{m})$, we can use a verification procedure almost identical to the procedure that appears in \cite{NeuSokAlgo}.  
However, if the uniform width, $\overline{m}$, is asymptotically smaller than the height of the tallest pattern, $m'$, then this algorithm does not yield a linear time text scanning.  This is due to the fact that the algorithm costs $O(m')$ time to process each candidate row, resulting in $O(m'*m')$ time if $\overline{m} =o(m')$.
For this situation, new ideas are needed and we introduce a new verification process that verifies a single pattern in $O(m')$ time.  Since the dictionary has $d$ patterns, and $d < \overline{m}$, the entire text block is verified in $O(\overline{m}m')$ time.

We verify candidates for each pattern, $P_i$, separately.  Verification of each candidate consists of two tasks:
\begin{enumerate}
	\item {\bf Verify shifts:} \label{verifyShift}
	Let $P'_i$ be the 1D pattern of names for $P_i$. If $P'_i$ is not periodic, there are $O(1)$ candidates in a text block, and we verify each candidate  for $P_i$ separately by matching $P_i$'s 2D Lyndon word with the 2D Lyndon word of the corresponding rows of the text block. 
	If $P'_i$ is periodic, the idea is similar. We call each period in $P'_i$ a $p\_block$. 
	We first verify the shifts within each $p\_block$ and then verify the shifts between adjacent $p\_blocks$.
	We compute the 2D Lyndon word of each $p\_block$ separately and store the column $z$ that it occurs in.  
	Since each $p\_block$ has the same horizontal period (i.e., the LCM of the periods of the rows of a $p\_block$),
	we use a Knuth-Morris-Pratt automaton \cite{KMP} on the  $p\_blocks$ to complete the verification. 
	The KMP automaton verifies that corresponding $p\_blocks$ have the same name and that
	the difference between the $z$ values of adjacent $p\_blocks$ is the same in the text and in the pattern.
	
\item	\textbf{Check width:} \label{verifyWidth}
Use range minimum and maximum queries to calculate \emph{minRight} and \emph{maxLeft} for each candidate of $P_i$.  Then, reverse the shift and make sure that there is room for the pattern between \emph{minRight} and \emph{maxLeft}, i.e., that the candidate spans at least $\overline{m}$ columns.
	
\end{enumerate}

\begin{restatable}{lemma}{lemmaGroupIscan}
\label{lemma:group1scan}
A text of size $n_1 \times n_2$ is searched for patterns in Group I in $O(n_1n_2 \tau)$ time and $O(\overline{m} \log \overline{m}+m' \log dm')$ bits of extra space.
\end{restatable}

\begin{proof}
The linear representation of the text block is computed in $O(\overline{m}m')$ time and occupies $O(m' \log dm')$ bits of space, as shown in Section~\ref{sec:group1pre}. 
Candidates are identified with Sahinalp and Vishkin's algorithm \cite{SahVis96} in time linear in the 1D representations. Verification  as done in \cite{NeuSokAlgo} is linear. 
It remains to show that the new verification, when $\overline{m}=o(m')$, runs in linear time.
Computing the 2D Lyndon word for the entire text block or for each of the $p\_block$s in the text block takes $O(m')$ time. KMP on the 2D Lyndon words of the $p\_block$s and the shifts between $p\_block$s takes $O(m')$ time.
Thus, $P_i$ is verified in $O(\overline{m}+m')$ time, and all $d$ patterns are verified in $O(dm')$ time, since $d<\overline{m}$.
Linear time and space preprocessing schemes allow us to answer range minimum and maximum queries in $O(1)$ time \cite{GabBenTar84}.
Check-width (Step \ref{verifyWidth}) consists of constant-time RMQ per candidate, which totals $O(m')$ time overall for $P_i$, and for all $P_i$, $1\leq i \leq d$, text scanning completes in $O(\overline{m}m')$ time.

\noindent Each block of text is searched in $O(\overline{m}m' \tau)$ time and $O(\overline{m} \log \overline{m}+m' \log dm')$ bits of extra space.
Thus, the entire text is searched for patterns in Group I in $O(n_1n_2 \tau)$ time and $O(\overline{m} \log \overline{m}+m' \log dm')$ bits of extra space.
\qed
\end{proof}

\subsection{Group II Patterns}\label{sec:group2alg}
Patterns in Group II have at least one aperiodic row or one row whose period is larger than $\overline{m}/4$.  
We assume that each pattern in this group has at least one aperiodic row.  The case of a pattern having a row that is periodic with period size between $\overline{m}/4$ and $\overline{m}/2$ is handled similarly, since each pattern can occur only $O(1)$ times per text block row.

For patterns in Group II, many different pattern rows can overlap in a text block row.  As a result, it is difficult to employ a succinct naming scheme to linearize the text block and find all occurrences of patterns in the text.  Instead, we use the aperiodic row of each pattern to filter the text block and identify a limited set of candidates for pattern occurrences.  We use dynamic dueling \cite{NeuSokAlgo} to eliminate inconsistent candidates within each text column.  Then, a single pass over the text suffices to verify all remaining candidates for pattern occurrences.

\subsubsection{Preprocessing Patterns} \label{sec:group2pre}

The following preprocessing steps are initially performed for each dictionary pattern in Group II and are later used upon arrival of a new pattern.

\begin{enumerate}
	\item \label{step:aperiodic} Locate first aperiodic row and preprocess for dynamic dictionary matching. 
	
	\item \label{step:naming} Name pattern rows using a single witness tree and store 1D patterns of names. 
	\item \label{step:DCST}Insert pattern rows to dynamic compressed suffix tree.
	\item \label{step:gensuffixtree} Construct dynamic suffix tree of 1D patterns. 
	\item \label{step:LCA} Preprocess witness tree and suffix tree for dynamic LCA.
\end{enumerate}

\begin{restatable}{lemma}{lemmaGroupIIpre}
\label{lemma:group2pre}
Patterns in Group II are preprocessed in $O(\ell \tau)$ time and $O(d\overline{m} \log d\overline{m}+dm' \log dm')$ bits of extra space.
\end{restatable}

\begin{proof} 
In Step \ref{step:aperiodic}, the period of a pattern row is computed in $O(\overline{m})$ time and $O(\overline{m} \log \overline{m})$ bits of extra space \cite{Lothaire05}.  At most, all pattern rows are examined, in $O(\ell)$ time and  $O(\overline{m} \log \overline{m})$ bits of extra space. 
Sahinalp and Vishkin's algorithm indexes these rows in $O(d \overline{m})$ time and $O(d \overline{m} \log  d \overline{m})$ bits of space \cite{SahVis96}. 
Step \ref{step:naming} names pattern rows by the witness tree in $O(\ell)$ time. 
By Section \ref{sec:dynamicDS}
, the dynamic witness tree of pattern rows occupies $O(dm' \log dm')$ bits of space.   A single witness tree suffices since all pattern rows are the same size. 
Step \ref{step:DCST} indexes the pattern rows in a dynamic compressed suffix tree in $O(\ell \tau)$ time. 
Step \ref{step:gensuffixtree} constructs the dynamic suffix tree in $O(dm')$ time and stores it in $O(dm' \log dm')$ bits of space \cite{ChoLam97}. 
In Step \ref{step:LCA}, linear time preprocessing prepares the dynamic suffix and witness trees for $O(1)$ time LCA queries \cite{ColHar05}. 
\qed
\end{proof}

\begin{corollary}\label{cor:group2add}
The dictionary is updated to add a new pattern of size $p \times \overline{m}$ to Group II in
$O(p\overline{m} \tau)$ time and $O(p\overline{m} \log d\overline{m}+p \log dm')$ bits of extra space.
\end{corollary}

\begin{restatable}{lemma}{lemmaGroupIIremove}
\label{lemma:group2remove}
A pattern in Group II of size $p \times \overline{m}$ is removed from the dictionary in $O(p\overline{m} \tau)$ time and eliminates $O(\overline{m} \log \overline{m}+p \log dm')$ bits of extra space the algorithm allocated for it. 
\end{restatable}

\begin{proof}
The following steps are performed to remove a pattern from Group II.  
The first aperiodic row of the pattern is removed from the 1D dictionary that has been preprocessed for dynamic dictionary matching in $O(\overline{m})$ time and deallocates $O(\overline{m} \log \overline{m})$ bits of space \cite{SahVis96}. 
The 1D representation of the pattern is deleted and it is removed from the suffix tree of 1D patterns in $O(p)$ time and deallocates $O(p \log dm')$ bits of space \cite{ChoLam97}. 
Each row of the pattern is removed from the compressed suffix tree in $O(p \overline{m} \tau)$ time.  
\qed
\end{proof}

\subsubsection{Text Scanning}\label{sec:group2scan} 

The text is searched for patterns in Group II in almost the same way as in the static algorithm \cite{NeuSokAlgo}.  
The only difference between the text scanning stage of the static algorithm and that of the dynamic algorithm lies in the method used to identify 1D pattern occurrences in the linearized text.  The Aho-Corasick automaton is not suitable for a dynamic dictionary since it is not updated efficiently.  Rather, we use Sahinalp and Vishkin's method for dynamic dictionary matching since it completes all preprocessing and searching tasks, including updating the dictionary, in linear time and space.  We summarize the text scanning and the complexity analysis in the following.

\noindent \textbf{Summary of Text Scanning}
\begin{enumerate}
\item \label{group2:identify} \textbf{Identify candidates:} 
Sahinalp and Vishkin's  1D dynamic dictionary matching algorithm finds occurrences of the first aperiodic row of the patterns.  It searches the text block, one row at a time, 

\item \label{group2:duelVert} \textbf{Duel vertically:}
\begin{enumerate}
	\item An LCP query between suffixes of the 1D patterns finds the number of rows that match in overlapping candidates.  An LCA query in the suffix tree of 1D patterns is performed to find a row of mismatch.
	\item We use an LCA query in the witness tree to find a witness of mismatch between rows of different names.  Then a single character in each pattern row is retrieved and compared.
\end{enumerate}

\item \label{group2:verify} \textbf{Verify candidates:} 
We verify one text block row at a time and mark positions at which a pattern row (1D name) is expected to begin. 
Duels eliminate horizontally inconsistent candidates.  A duel 
 consists of an LCP query in the dynamic compressed suffix tree.  
After duels are performed, the surviving labels are carried to the next row.  
\end{enumerate}

\begin{restatable}{lemma}{lemmaGroupIIscan}
\label{lemma:group2scan}
A text of size $n_1 \times n_2$ is searched for patterns in Group II in $O(n_1n_2 \tau)$ time and $O(\overline{m} \log \overline{m}+dm' \log dm')$ bits of extra space.
\end{restatable}

\begin{proof}
Step \ref{group2:identify} searches each text block row for a single row of each pattern 
in $O(\overline{m}m')$ time and $O(\overline{m} \log \overline{m})$ bits of extra space \cite{SahVis96}.  
$O(d m')$ candidates are stored in $O(dm' \log dm')$ bits of extra space.
In Step  \ref{group2:duelVert}, each vertical duel consists of an $O(1)$-time LCP query in the suffix tree, an $O(1)$-time LCA query in the witness tree, and an $O(\tau)$-time character retrieval and comparison in a pair of pattern rows.
Overall, each duel takes $O(\tau)$ time.  Due to transitivity, the number of duels is limited by the number of candidates.
Since there are $O(dm')$ candidate positions, and $d < \overline{m}$,  the vertical duels complete in $O(\overline{m} m'  \tau)$ time.
In  Step \ref{group2:verify}, an LCP query in the dynamic compressed suffix tree  takes $O(\tau)$ time.
By transitivity,  the number of duels is limited by the number of candidates, which are $O(dm')$.  Since $d < \overline{m}$, dueling is completed in $O(\overline{m} m' \tau)$ time.
Verification uses space proportional to the labels for one text block row plus the number of candidates, $O(\overline{m} \log \overline{m}+ dm' \log dm')$ bits. 
Each text character within an anticipated pattern occurrence is only compared to one pattern character, in $O(\tau)$ time, which takes $O(\overline{m} m' \tau)$ time overall.  

\noindent Each block of text is searched in $O(\overline{m} m' \tau)$ time and $O(\overline{m} \log \overline{m}+dm' \log dm')$ bits of extra space.
Thus, the entire text is searched for patterns in Group II in $O(n_1n_2 \tau)$ time and $O(\overline{m} \log \overline{m}+dm' \log dm')$ bits of extra space.
\qed
\end{proof}

\begin{restatable}{theorem}{thmComplexityDlessM}
Our algorithm for dynamic 2D dictionary matching when $d < \overline{m}$ completes in $O((\ell +n_1n_2) \tau )$ time and $O(d\overline{m} \log d\overline{m}+dm' \log dm')$ bits of extra space.
Pattern $P$ of size $p \times \overline{m}$ can be inserted to or removed from the dictionary in $O(p\overline{m}\tau)$ time and the index will occupy an additional $O(p \log dm')$ bits of space, where $m'$ is updated to reflect the new maximum pattern height.
\end{restatable}

\begin{proof}
We separate the patterns into two groups and search for patterns in each group separately.  Classifying a pattern entails finding the period of each pattern row.  This is done in $O(\overline{m})$ time and $O(\overline{m} \log \overline{m})$ bits of extra space per row \cite{Lothaire05}.  Overall, the dictionary is separated into two groups in $O(\ell)$ time and $O(\overline{m} \log \overline{m})$ bits of extra space.   
For patterns in Group I, this complexity is demonstrated by Lemmas \ref{lemma:group1pre}, \ref{lemma:group1remove}, \ref{lemma:group1scan} and Corollary \ref{cor:group1add}.
For patterns in Group II, this complexity is demonstrated by Lemmas \ref{lemma:group2pre},  \ref{lemma:group2remove}, \ref{lemma:group2scan} and Corollary \ref{cor:group2add}.
\qed
\end{proof}

\section{Conclusion}\label{sec:conclusion}
We have presented the first efficient dynamic 2D dictionary matching algorithm that runs in sublinear working space.  The algorithm is a succinct and dynamic version of the classic Bird / Baker algorithm.  Since we follow their labeling paradigm, our algorithm 
is suited for a dictionary of rectangular patterns that are all the same size in at least one dimension.
Our algorithm uses a dynamic compressed suffix tree as a compressed self-index to represent the dictionary in entropy-compressed space.  All tasks are completed by our algorithm in linear time, overlooking the slowdown in querying the compressed suffix tree.  

When the rectangular patterns are of different height, width and aspect ratios, a method that labels text positions is not appropriate.
Idury and Schaffer developed a dynamic dictionary matching algorithm for such patterns \cite{IduSch95}.  Their algorithm uses techniques for multidimensional range searching as well as several applications of the Bird / Baker algorithm, after splitting each pattern into overlapping pieces and handling these segments in groups of uniform height.  
 Idury and Schaffer's algorithm requires working space proportional to the dictionary size.  
We hope that our succinct dynamic version of the Bird / Baker algorithm is a first step towards addressing the more general problem of succinct dynamic 2D dictionary matching among all rectangular patterns.

Many problems related to succinct dynamic dictionary matching remain open.  In future work we hope to address succinct 2D dictionary matching when the pattern occurrences can be approximately matched to the text.  The approximate matches may accommodate character mismatches, insertions, deletions, ``don't care'' characters, or swaps.



\end{document}